\documentclass[12pt,reqno]{amsart}
\usepackage[left=3.2cm, top=3.4cm, bottom=3.4cm, right=3.2cm]{geometry}
\usepackage{mathrsfs}
\usepackage{amsmath,amsxtra,amsfonts,amssymb, amsthm, amscd, epsfig}
\usepackage[dvipsnames,usenames]{color}
\usepackage{amssymb,amsmath,amsfonts,epsfig}
\usepackage{graphicx}

\numberwithin{equation}{section}

\theoremstyle{plain}
\newtheorem{thm}{Theorem}[section]
\newtheorem{Mthm}{Main Theorem}
\newtheorem{prop}[thm]{Proposition}
\newtheorem{cor}[thm]{Corollary}

\newtheorem{lem}[thm]{Lemma}
\newtheorem{de}[thm]{Definition}
\newtheorem{rem}[thm]{Remark}
\newtheorem{ex}[thm]{Example}
\newcommand{\eqa}{\begin{eqnarray}}
\newcommand{\eeqa}{\end{eqnarray}}
\newcommand{\beq}{\begin{equation}}
\newcommand{\eeq}{\end{equation}}
\newcommand{\nn}{\nonumber}
\newcommand{\p}{\partial}

\def \la {\langle}
\def \ra{\rangle}
\def \var{\varepsilon}

\def \dsum{\displaystyle\sum}

\def \trf{\omega \,}

\begin{document}

\title[]
{Frobenius manifolds and Frobenius algebra-valued
integrable systems}
\author[]{Ian A.B. Strachan and Dafeng Zuo}

\address{Ian A.B. Strachan,School of Mathematics and Statistics, University of Glasgow}
\email{ian.strachan@glasgow.ac.uk}

\address[]{Dafeng Zuo,School of Mathematical Science,University of Science and
Technology of China, Hefei 230026, P.R.China and Wu Wen-Tsun Key Laboratory of Mathematics,
USTC, Chinese Academy of Sciences and School of Mathematics}

\email{dfzuo@ustc.edu.cn}

\date{\today}

\begin{abstract}
The notion of integrability will often extend from
systems with scalar-valued fields to systems with algebra-valued fields. In such extensions the properties of,
and structures on, the algebra play a central role in ensuring integrability is preserved. In this paper a new theory of Frobenius-algebra valued integrable systems
is developed.

This is achieved for systems derived from Frobenius manifolds by utilizing the theory of tensor products for such manifolds, as developed by Kaufmann, Kontsevich and Manin \cite{Kaufmann,KMK}.
By specializing this construction, using a fixed Frobenius algebra $\mathcal{A},$ one can arrive at such a theory. More generally one can apply the same idea to construct an $\mathcal{A}$-valued Topological Quantum Field Theory.

The Hamiltonian properties of two classes of integrable evolution equations are then studied: dispersionless and dispersive evolution equations. Application of these ideas are discussed and, as an example, an $\mathcal{A}$-valued modified Camassa-Holm equation is constructed.

\end{abstract}



\maketitle 
{ \tableofcontents}

\section{Introduction}

Of the many ways to generalize the Korteweg-de Vries equation $u_t=u_{xxx}+6 u u_x\,,$ the one
that will be of most relevance to this
paper is the matrix generalization (see, for example, \cite{BI-1, BI-2})
\begin{equation}
\mathcal{U}_t = \mathcal{U}_{xxx} + 3 \mathcal{U} \mathcal{U}_x+ 3 \mathcal{U}_x \mathcal{U}\,,
\label{matrixKdV}
\end{equation}
where the two first derivative terms are required due to the non-commutativity of matrix
multiplication. If one restricts such an equation to
the space of commuting matrices one arrives at the equation $\mathcal{U}_t =
\mathcal{U}_{xxx} + 6 \mathcal{U} \mathcal{U}_x$ which is identical in form to the original
KdV equation but with a matrix-valued, as opposed to a scalar-valued, field (see, for example, \cite{Kup2000,Zuo-2013,Zuo-2014-1}).
The purpose of this paper is to construct $\mathcal{A}$-valued, where $\mathcal{A}$ is a Frobenius algebra, generalizations of integrable
systems, starting with those associated to an underlying Frobenius manifold and related dispersionless hierarchies, and extending the ideas to topological quantum field theories
and dispersive hierarchies.

\medskip

The structure of this paper may be summarized in the following diagram:

\[
\begin{array}{ccc}
\left\{
\begin{array}{c}
\mathcal{A}-{\rm valued}\\
{\rm Frobenius~manifold~(\S 2)}
\end{array}
\right\}
& \longrightarrow &
\left\{
\begin{array}{c}
\mathcal{A}-{\rm valued}\\
{\rm TQFT~(\S 3)}
\end{array}
\right\}\\
&&\\
\Big\downarrow && \\
&&\\
\left\{
\begin{array}{c}
\mathcal{A}-{\rm valued~bi-Hamiltonian}\\
{\rm dispersionless~systems~(\S 4)}
\end{array}
\right\}
& \longrightarrow &
\left\{
\begin{array}{c}
\mathcal{A}-{\rm valued~bi-Hamiltonian}\\
{\rm dispersive~systems~(\S 5)}
\end{array}
\right\} \\
\end{array}
\]
The full reconstruction of a dispersive hierarchy (the missing vertical arrow in the above diagram) remains an open problem, even before
one considers $\mathcal{A}$-valued systems.

\medskip

The starting point (Section 2) for the study of such $\mathcal{A}$-valued hierarchies is the classical construction of
Dubrovin \cite{Du} which associates to a Frobenius manifold a bi-Hamiltonian hierarchy of hydrodynamic type.
By constructing the tensor product \cite{Kaufmann,KMK} of such a manifold with a trivial Frobenius manifold
(i.e. a fixed algebra) one automatically obtains a new Frobenius manifold and hence a bi-Hamiltonian hierarchy.
The component fields of this new hierarchy can then be reassembled to form an $\mathcal{A}$-valued hierarchy.
The important feature of this construction is a simple, explicit, form of the new prepotential that defines the
$\mathcal{A}$-valued hierarchies.

\medskip

More explicitly, given a Frobenius algebra $\mathcal{A}$ with basis $e_i\,,i=1\,,\ldots\,,n\,,$ one can replace the flat
coordinates of a Frobenius manifold $\mathcal{M}$ with $\mathcal{A}$-valued fields via the map
\[
\hat{}\,:\,t^\alpha \mapsto \widehat{t^\alpha} = t^{(\alpha i)} e_i,\quad\alpha=1\,,\ldots \,, m,\quad i=1\,,\ldots\,,n\nn
\]
and this action can be extended to functions, at least in the case of analytic Frobenius manifolds (and to wider classes of functions - see the Appendix). Conversely, an $\mathcal{A}$-valued
field can be reduced to a scalar field via the Frobenius form (or trace form) $\omega\,.$ This construction is described in Section \ref{section2}.
The main result is the following:

\begin{Mthm}{\rm(Theorem \ref{MainA})}
{Let $F$ be the prepotential of a Frobenius manifold $\mathcal{M}$
and let $\mathcal{A}$ be a trivial Frobenius algebra with 1-form $\omega\,.$ The
function
\[
{F}^\mathcal{A} = \omega\left(\widehat{F}\right)
\]
defines a Frobenius manifold, namely the manifold $\mathcal{M}\otimes \mathcal{A}\,.$ }\end{Mthm}
\medskip

\noindent Normally the prepotential of a tensor product of Frobenius manifolds bears little resemblance to the underlying prepotentials,
and in any case is only defined implicitly from the original prepotentials. However when one of the manifolds is trivial, the above closed form of the
new prepotential exists and this enables the resulting hierarchies to be constructed explicitly.

\medskip

In Section \ref{ATQFT} we extend these ideas to a full Topological Quantum Field Theory on the big phase space $\mathcal{M}^\infty$, i.e. with gravitational
descendent fields.

\begin{Mthm} {\rm(Theorem \ref{MthmB})} Let $\mathcal{F}_{g\geq 0}$ be the prepotentials defining a TQFT, $\mathcal{S}$ and $\mathcal{D}$ the corresponding
String and Dilaton vector fields and $\mathcal{A}$ be a trivial Frobenius algebra. Let $f$ be an analytic function on
$\mathcal{M}^\infty$ and define the $\mathcal{A}$-valued
function $\hat{f}$ to be:
\begin{equation}
{\hat{f}} = \left. f\right|_{t^\alpha_N \mapsto t^{(\alpha i)}_N e_i}\,,\qquad N\in\mathbb{Z}_{\geq 0}\,,\quad\alpha=1\,,\ldots \,, m,\quad i=1\,,\ldots\,,n\,.\nn
\end{equation}
Then the functions
\[
\mathcal{F}^{\mathcal{A}}_{g\geq 0} = \omega\left( {\widehat{ \mathcal{F}}_{g\geq 0} }\right)
\]
and vector fields
\begin{eqnarray*}
\mathcal{S}^\mathcal{A} = -\sum_{N,(\alpha i)} \tilde{t}^{(\alpha i)}_N \tau_{N-1,(\alpha i)},\quad
\mathcal{D}^\mathcal{A} = -\sum_{N,(\alpha i)} \tilde{t}^{(\alpha i)}_N \tau_{N,(\alpha i)}
\end{eqnarray*}
satisfy the axioms of a Topological Quantum Field Theory.
\end{Mthm}
In the remaining sections a theory of $\mathcal{A}$-valued integrable systems is developed, first for dispersionless
systems and then for certain dispersive systems.
More specifically, in section \ref{dispersionlesssection} the construction of the $\mathcal{A}$-valued dispersionless (or hydrodynamic) hierarchies is given. The deformed flat coordinates can be
described very simply, and these form the Hamiltonian densities for the new evolution equations. By reassembling the fields these equations
can be written as $\mathcal{A}$-valued evolution equations. To write these in Hamiltonian form requires the definition of a functional
derivative with respect to an $\mathcal{A}$-valued field, and such a derivative was defined in \cite{OS} and with this one can write the flow equations as
$\mathcal{A}$-valued bi-Hamiltonian evolution equations. These ideas are then extended to the dispersive case in section \ref{dispersivesection}.

\medskip

\section{Frobenius manifolds and their tensor products}\label{section2}

\subsection{Frobenius algebras and manifolds}

We begin with the definition of a Frobenius algebra \cite{Du}.

\begin{de}
A Frobenius algebra $\{ \mathcal{A},\circ,e,\omega\}$ over $\mathbb{R}$ satisfies the following conditions:

\begin{itemize}

\item[{\sl (i)}] $\circ\,: \mathcal{A} \times \mathcal{A} \rightarrow \mathcal{A}$ is
a commutative, associative algebra with unity $e$;

\item[{\sl (ii)}] $\omega \in \mathcal{A}^\star$ defines a non-degenerate inner product
$\langle a,b \rangle = \omega(a \circ b)\,.$

\end{itemize}

\end{de}

\noindent Since $\omega(a) = \langle e,a \rangle$ the inner product determines the form
$\omega$ and visa-versa. This linear form $\omega$ is often called a trace form (or Frobenius form).
One dimensional Frobenius algebras are trivial: the requirement of an identity and the non-degeneracy
of the inner product determines the algebra uniquely and the inner product up to a non-zero constant.
Two dimensional algebra are easily classified.

\begin{ex} \label{ex2.2}
Let $\mathcal{A}$ be a $2$-dimensional commutative and associative algebra with a basis $e=e_1, e_2$
satisfying
\beq e_1\circ e_1=e_1,\quad e_1\circ e_2=e_2, \quad e_2\circ e_2=\var e_1+\mu e_2,\quad \var,\mu \in \mathbb{R}.
 \eeq
Obviously, the algebra $\mathcal{A}$ has a matrix
representation as follows
\beq e_1 \mapsto \mathrm{I}_2=\left(\begin{array}{cc}
1&0\\
0&1
\end{array}\right),\quad e_2 \mapsto\left(\begin{array}{cc}
0& {\var}\\
1&\mu
\end{array}\right).\nn \eeq
It is easy to show that:
\begin{itemize}

\item[(1)] if $\mu^2=-4\var$, $\mathcal{A}$ is nonsemisimple, i.e.,
$\exists \, \widetilde{e}=\mu e_1-2 e_2$ such that
$\widetilde{e}\circ\widetilde{e}=0$;
\item[(2)] if $\mu^2\ne-4\var$, then $\mathcal{A}$ is semisimple, i.e., for
any nonzero element $\widetilde{e}=x e_1+y e_2$, $\widetilde{e}\circ\widetilde{e}\ne 0$.
\end{itemize}
Furthermore,  we introduce
two ``basic" trace-type forms for $a=a_1e_1+a_2e_2\in\mathcal{A}$ as follows
\beq \omega_k(a)=a_k+a_2(1-\delta_{k,2})\delta_{\var,0}, \quad k=1,\, 2,\label{ZZ2.4}\eeq
which induce two nondegenerate inner products on $\mathcal{A}$ given by
 \beq \la a, b\ra_k:=\omega_k(a\circ b),\quad a,\,b\in \mathcal{A},\quad k=1,2.\label{SZ2.6}\eeq
The two Frobenius algebras
$\left\{\mathcal{A},\circ, e, \omega_k\right\}$ will be denoted by $\mathcal{Z}_{2,k}^{\var,\mu}$ for $k=1,2$.
\end{ex}

\begin{ex}\label{ex2.3}

Let $\mathcal{A}$ be an $n$-dimensional  nonsemisimple commutative associative
algebra $\mathcal{Z}_n$ over $\mathbb{R}$ with a unity $e$ and a basis $e_1=e, \cdots,
e_n$ satisfying
\beq e_i \circ e_j=\left\{
\begin{array}{ll}
e_{i+j-1}, & i+j\leq n+1,\\
0,&  i+j=n+2. \end{array}\right.
 \eeq
Taking $\Lambda=(\delta_{i,j+1})\in gl(m,\mathbb{R})$, one obtains
a matrix representation of $\mathcal{A}$ as
 \beq e_j \mapsto \Lambda^{j-1}, \quad j=1,\cdots, n.\nn \eeq
Similarly, for any $a=\dsum_{k=1}^{n} a_k e_k\in \mathcal{A}$,
we introduce $n$ trace-type forms, called ``basic" trace-type forms,  as follows
\beq \omega_{k-1} (a)=a_k+a_{n}(1-\delta_{k,n}),\quad k=1,\cdots,n.\label{ZZ2.3} \eeq
Every trace map $\omega_k$ induces a nondegenerate
symmetric bilinear form on $\mathcal{A}$ given by
 \beq \la a, b\ra_k:=\omega_k (a\circ b),\quad a,\,b\in \mathcal{A},\quad k=0,\cdots,n-1.\label{AZ2.6}\eeq
Thus all of $\left\{\mathcal{A}, \circ, e, \omega_{k-1}\right\}$ are nonsemisimple
Frobenius algebras, denoted by $\mathcal{Z}_{n,k-1}$ for $k=1,\cdots,n$.
We remark that if we consider
a linear combination of $n$ ``basic" trace-type forms as
\beq \mathrm{tr}_n:=\dsum_{s=0}^{n-1}\omega_{s}-(n-1)\,\omega_{n-1},\nn\eeq
then $\left\{\mathcal{A}, \circ, e, \mathrm{tr}_n\right\}$ is also a Frobenius
algebra which is exactly  the algebra $\left\{\mathcal{Z}_n,\mathrm{tr}_n\right\}$
used in \cite{Zuo-2013}\footnote{It was the realization that the matrix algebra used in this paper was a specific example of
a Frobenius algebra that led to the development of the current paper.}.

\end{ex}

A Frobenius manifold has such a structure on each tangent space.

\begin{de} \cite{Du} The set $\{\mathcal{M},\circ,e,\langle~,~\rangle,E \}$ is a Frobenius manifold if
each tangent space $T_t\mathcal{M}$ carries a smoothly varying Frobenius algebra
with the properties:

\begin{itemize}

\item[({\sl i})] $\langle~,~\rangle$ is a flat metric on $\mathcal{M}$;

\item[{\sl (ii)}] $\nabla e=0$, where $\nabla$ is the Levi-Civita connection of $\langle\,,\rangle$;

\item[{\sl (iii)}] the tensors $c(u,v,w):=\langle u\circ v,w\rangle$ and $\nabla_zc(u,v,w)$ are totally symmetric;

\item[{\sl (iv)}] A vector field $E$ exists, linear in the flat-variables, such that the
corresponding group of diffeomorphisms acts
by conformal transformation on the metric and by rescalings on the algebra on $T_t\mathcal{M}\,.$

\end{itemize}

\end{de}

\noindent These axioms imply the existence of the prepotential $F$ which satifies the
WDVV-equations of associativity in the flat-coordinates of
the metric (strictly speaking only a complex, non-degenerate bilinear form) on $\mathcal{M}\,.$
The multiplication is then defined by the third derivatives
of the prepotential:
\[
\frac{\partial~}{\partial t^\alpha} \circ \frac{\partial~}{\partial t^\beta}
= c_{\alpha\beta}^{\phantom{\alpha\beta}\gamma}({\bf t}) \frac{\partial~}{\partial t^\gamma}
\]
where
\[
c_{\alpha\beta\gamma} = \frac{ \partial^3 F}{\partial t^\alpha \partial t^\beta \partial t^\gamma}
\]
and indices are raised and lowered using the metric $\eta_{\alpha\beta}
= \langle \dfrac{\partial~}{\partial t^\alpha},\dfrac{\partial~}{\partial t^\beta} \rangle$.

\begin{ex}\label{trivial} Suppose $c_{ij}^{~~k}$ are the structure constants
for the Frobenius algebra $\mathcal{A}$, so $e_i \circ e_j = c_{ij}^{~~k} e_k$
and $\eta_{ij} = \langle e_i,e_j \rangle\,.$ For such an algebra
one obtains a cubic prepotential
\begin{eqnarray*}
F & = & \frac{1}{6} c_{ijk} t^i t^j t^k\,,\\
& = & \frac{1}{6} \omega( {\bf t}\circ{\bf t}\circ{\bf t})\,,\qquad {\bf t} = t^i e_i\,.
\end{eqnarray*}
The Euler vector field takes the form $E=\dsum_{i} t^i\frac{\p}{\p t^i}$ and $E(F)=3F\,.$
The notation $\mathcal{A}$ will be used for {\sl both} the algebra and the corresponding manifold.
\end{ex}

Motivated by the classical K\"unneth formula in cohomology, Kaufmann, Kontsevich and Manin \cite{Kaufmann, KMK} constructed the tensor product of two
Frobenius manifolds $\mathcal{M}^\prime$ and $\mathcal{M}^{\prime\prime}$, denoted $\mathcal{M}^\prime \otimes\mathcal{M}^{\prime\prime}\,.$
The following formulation of this construction is taken from \cite{Du2}. This formulation also gives criteria to check if a particular manifold is the tensor product of two more basic manifolds. For simplicity we use the notation
$\p_\alpha=\dfrac{\p}{\p t^\alpha}$ and $\p_{\alpha\beta}=\dfrac{\p~~}{\p t^{(\alpha\beta)}}$.

\begin{prop}\label{tensor} Let $\mathcal{M}^\prime$ and $\mathcal{M}^{\prime\prime}$
be two Frobenius manifolds of dimension $n^\prime$ and $n^{\prime\prime}\,.$
A Frobenius manifold $\mathcal{M}$ of dimension $n^\prime n^{\prime\prime}$ is the
tensor product $\mathcal{M} = \mathcal{M}^\prime\otimes\mathcal{M}^{\prime\prime}$ if the
following conditions hold:

\begin{itemize}

\item[({\sl i})] $\{ T\mathcal{M}, \langle\,,\rangle,e\}
= \{ T\mathcal{M}^\prime\otimes T\mathcal{M}^{\prime\prime}, \langle\,,
\rangle^\prime\otimes\langle\,,\rangle^{\prime\prime},
e^\prime \otimes e^{\prime\prime} \}\,.$
Flat coordinates are labeled by pairs $t^{(\alpha^\prime\alpha^{\prime\prime})}\,,\alpha^\prime
=1\,,\ldots\,,n^\prime\,,\alpha^{\prime\prime} =1\,,\ldots\,,n^{\prime\prime}\,,$
and the unity vector field is
\[
e=\dfrac{\p~}{\partial t^{(1 1)}}
\]
and the metric $\langle\,,\rangle$ has the form
\[
\eta_{(\alpha^\prime\alpha^{\prime\prime})(\beta^\prime\beta^{\prime\prime})}
= \eta_{\alpha^\prime\beta^\prime} \, \eta_{\alpha^{\prime\prime} \beta^{\prime\prime}}\,.
\]

\item[({\sl ii})] At a point $t^{(\alpha^\prime\alpha^{\prime\prime})}
=0\,,\alpha^\prime>1\,,\alpha^{\prime\prime}>1$ the algebra $T_t\mathcal{M}$ is a tensor
product
\[
T_t\mathcal{M} = T_{t^\prime} \mathcal{M}^\prime \otimes T_{t^{\prime\prime}} \mathcal{M}^{\prime\prime}\,,
\]
that is:
\[
c_{(\alpha^\prime\alpha^{\prime\prime})(\beta^\prime\beta^{\prime\prime})}^{\phantom{(\alpha^\prime
\alpha^{\prime\prime})(\beta^\prime\beta^{\prime\prime})}(\gamma^\prime\gamma^{\prime\prime})}(t)=
c_{\alpha^\prime\beta^\prime}^{\phantom{{\alpha^\prime\beta^\prime}}\gamma^\prime}(t^\prime) \,
c_{\alpha^{\prime\prime}\beta^{\prime\prime}}^{\phantom{{\alpha^{\prime\prime}\beta^{\prime\prime}}}
\gamma^{\prime\prime}}(t^{\prime\prime})\,.
\]

\item[({\sl iii})] If the Euler vector fields of the two manifolds $\mathcal{M}$ and
$\mathcal{M}^{\prime\prime}$ take the form

\begin{eqnarray*}
E^\prime & = &  \sum_{\alpha^\prime} \left[ (1-q_{\alpha^\prime}) t^{\alpha^\prime}
+ r_{\alpha^\prime}\right] \partial_{\alpha^\prime}\,,\\
E^{\prime\prime} & = &  \sum_{\alpha^{\prime\prime}} \left[ (1-q_{\alpha^{\prime\prime}})
t^{\alpha^{\prime\prime}} + r_{\alpha^{\prime\prime}}\right] \partial_{\alpha^{\prime\prime}}\,,\\
\end{eqnarray*}
with scaling dimensions $d^\prime$ and $d^{\prime\prime}$ respectively, then the
 Euler vector field on $\mathcal{M}$ takes the form
\[
E=\sum_{\alpha^\prime,\alpha^{\prime\prime}} (1-q_{\alpha^\prime}-q_{\alpha^{\prime\prime}})
\partial_{(\alpha^\prime\alpha^{\prime\prime})} +
\sum_{\alpha^\prime} r_{\alpha^\prime}\partial_{\alpha^\prime 1^{\prime\prime}} +
\sum_{\alpha^{\prime\prime}} r_{\alpha^{\prime\prime}} \partial_{{1^\prime \alpha^{\prime\prime}}}
\]
and $d=d^\prime+d^{\prime\prime}\,.$

\end{itemize}

\end{prop}

\noindent Such products describe the quantum cohomology of a product of varieties,
 and within singularity theory it appears when one takes the
direct sum of singularities.

\subsection{Tensor products with trivial algebras}

We now take the tensor product of a Frobenius manifold $\mathcal{M}$ with a trivial
 manifold $\mathcal{A}$ defined by a Frobenius algebra (Example \ref{trivial}). To
emphasize the different roles played by $\mathcal{M}$ and $\mathcal{A}$ we alter
 the general notation for tensor products as described above. The tensor product
will be written as $\mathcal{M}_{\mathcal{A}}\,,$ (so  $\mathcal{M}_{\mathcal{A}}
= \mathcal{M}\otimes{\mathcal{A}}$). The basis $e_i$ for $\mathcal{A}$ will be retained and the
unity element denoted by $e_1\,.$ Thus notation such as $e=\partial_{1}$ will not be used.
Latin indices will be reserved for $\mathcal{A}$-related objects, and Greek indices will be
reserved for $\mathcal{M}$-related objects. Thus $c_{\alpha\beta}^{~~\gamma}$
will denote the structure functions for the multiplication on $\mathcal{M}$ and $c_{ij}^{~~k}$ will
denote the structure constants for the multiplication on $\mathcal{A}\,.$
Coordinates on $\mathcal{M}_{\mathcal{A}}$ are denoted
\[
\{ t^{(\alpha i)}\,, \alpha = 1 \,,\ldots\,, m=dim \mathcal{M}\,,\quad i=1\,,\ldots\,, n=dim \mathcal{A} \}\,.
\]
No confusion should arise with this notation.

We begin by constructing a lift of a scalar valued function to
an $\mathcal{A}$-valued function and visa-versa.

\begin{de}\label{LiftingDef} Let $f$ be an analytic function on $\mathcal{M}$
(that is, analytic in the flat coordinates for $\mathcal{M}$). The $\mathcal{A}$-valued
function $\hat{f}$ is defined to be:
\[
{\hat{f}} = \left. f\right|_{t^\alpha \mapsto t^{(\alpha i)} e_i}
\]
with $\widehat{fg}=\hat{f}\circ\hat{g}$ and $\hat{1}=e_1\,.$
The evaluation $f^\mathcal{A}$ of $\hat{f}$ is defined by
\[
f^\mathcal{A}=\omega\left(\hat{f}\right)\,,
\]
where $\omega\in\mathcal{A}^\star\,.$

\end{de}

\noindent Since the function is analytic and the algebra $\mathcal{A}$ is
commutative and associative the above construction is well-defined.

\begin{rem} This definition requires the existence of a distinguished
coordinate system on $\mathcal{M}$ in which the function $f$ is
analytic. In the case of analytic Frobenius manifolds one automatically
 has such a distinguished system of coordinates, namely the
flat coordinates of the metric.
\end{rem}

\noindent With these definitions one may construct a new prepotential from the original one.

\begin{thm}\label{MainA} Let $F$ be the prepotential of a Frobenius manifold $\mathcal{M}$
and let $\mathcal{A}$ be a  Frobenius algebra with 1-form $\omega\,.$ The
function
\[
F^\mathcal{A} = \omega\left(\widehat{F}\right)
\]
defines a Frobenius manifold, namely the manifold $\mathcal{M}_\mathcal{A}\,.$

\end{thm}

\noindent Note, one could \lq straighten out\rq~ the coordinates $t^{(\alpha i)}$ via the map
$$ v^{i+(\alpha-1) n}=t^{(\alpha i)}\,, \quad 1\leq i \leq n,\quad 1\leq \alpha\leq m,$$
and hence $F^\mathcal{A}=F^\mathcal{A}(v^1,\cdots,v^{mn})$. However such a map is not unique
and the tensor structure
is lost.

\begin{proof} The proof is in two parts: we first show that the
prepotential $F^\mathcal{A}$ defines a Frobenius manifold, and then identify this
with the tensor product $\mathcal{M} \otimes \mathcal{A}\,.$

By construction we have an $nm$-dimensional manifold with coordinates
 $t^{(\alpha i)}\,,\alpha=1\,,\ldots\,,m=dim\mathcal{M}\,,i=1\,,\ldots\,,n=dim\mathcal{A}\,.$ We
begin with two simple results:

\begin{itemize}

\item[$\bullet$] Because $\eta_{ij} = \omega(e_i \circ e_j)$ it follows,
since by definition, $(\eta^{ij})=(\eta_{ij})^{-1}\,,$ that
\[
\omega(e_i \circ e_r) \eta^{rs} \omega(e_s\circ e_j) = \omega(e_i\circ e_j)\,.
\]
More generally, using the properties of the multiplication on $\mathcal{A}\,,$
\begin{equation}\label{propertyA}
\omega(\ldots \circ e_i \circ e_r) \eta^{rs} \omega(e_s\circ e_j\circ \ldots)
 = \omega(\ldots \circ e_i\circ e_j \circ \ldots)\,.
\end{equation}

\item[$\bullet$] The fundamental result that will be used extensively in the rest of the paper is the following:
\begin{equation}\label{propertyB}
\frac{\partial\hat{f}}{\partial t^{(\alpha i)}}
= \widehat{\frac{\partial f}{\partial t^\alpha}} \circ e_i\,.
\end{equation}
We introduce the notation $\hat{f} = [\hat{f}]^p e_p\,,$ so
\[
\frac{\partial\hat{f}}{\partial t^{(\alpha i)}}
= \left[\widehat{\frac{\partial f}{\partial t^\alpha}}\right]^p e_p \circ e_i\,.
\]
This will be used to separate out the $\mathcal{A}$-valued part of various expressions.

\end{itemize}

\noindent With these,
\[
\frac{\partial^3 \hat{F}}{\partial t^{(\alpha i)} \partial t^{(\beta j)} \partial t^{(\gamma k)}}
 = \widehat{\left( \frac{\partial^3 F}{\partial t^\alpha \partial t^\beta \partial t^\gamma}\right)}
\circ e_i \circ e_j \circ e_k\,,
\]
so
\begin{eqnarray*}
\frac{\partial^3 F^\mathcal{A}}{\partial t^{(\alpha i)} \partial t^{(\beta j)} \partial t^{(\gamma k)}}
& = & \omega\left( \widehat{c_{\alpha\beta\gamma}}
\circ e_i \circ e_j \circ e_k\right)\,,\\
& = & \left[ \widehat{c_{\alpha\beta\gamma}} \right]^p \omega(e_p\circ e_i \circ e_j \circ e_k)\,,\\
& = & c_{(\alpha i)(\beta j)(\gamma k)}\,.
\end{eqnarray*}

\noindent{\underline{Normalization}}

We define $\eta_{(\alpha i)(\beta j)}$ by
\begin{eqnarray*}
\eta_{(\alpha i)(\beta j)} & = & c_{(11)(\alpha i)(\beta j)}\,,\\
& = & \omega\left( \widehat{c_{1\alpha\beta}}
\circ e_1 \circ e_i \circ e_j\right)\,,\\
& = & \eta_{\alpha\beta}\, \eta_{ij}
\end{eqnarray*}
since $\widehat{c_{1\alpha\beta}} = \widehat{\eta_{\alpha\beta}}
= \eta_{\alpha\beta} e_1\,,$ and $e_1$ is the unity for the multiplication on $\mathcal{A}\,.$

This is non-degenerate (since by assumption $\eta_{\alpha\beta}$
and $\eta_{ij}$ are non-degenerate) and this will be taken to be
 the metric and used to raise and lower indices.
In particular,
$\eta^{(\alpha i)(\beta j)}=\eta^{\alpha\beta}\, \eta^{ij}\,.$

\noindent{\underline{Associativity}}

Using the metric to raise an index one obtains
\begin{equation}\label{liftedmultiplication}
c_{(\alpha i)(\beta j)}^{\phantom{{(\alpha i)(\beta j)}}(\gamma k)}
= \left[ \widehat{c_{\alpha\beta}^{~~\gamma}} \right]^p \, c_{ij}^{~~q} c_{pq}^{~~k}
\end{equation}
and this defines a multiplication on $\mathcal{M}_{\mathcal{A}}\,.$
The structure of this multiplication may be made more transparent if one writes the
basis for $T\mathcal{M}_{\mathcal{A}}$ as a tensor product:
\[
\frac{\partial~}{\partial t^{(\alpha i)}} = \partial_\alpha \otimes e_i\,.
\]
With this, the multiplication may be written as:
\[
\left( \partial_\alpha \otimes e_i\right) \circ \left( \partial_\beta \otimes e_j\right)
= \left[ \widehat{ \partial_\alpha \circ \partial_\beta }\right]^p \otimes e_p \circ e_i\circ e_j\,,
\]
where $ \widehat{f\partial_\alpha}=[\hat{f}]^p \partial_\alpha\otimes e_p\,,$ and
hence $\left[\widehat{f\partial_\alpha}\right]^p=[\hat{f}]^p \partial_\alpha\,.$
By construction this multiplication defines a commutative multiplication with unity
$e=\dfrac{\p~}{\partial t^{(1 1)}}=\partial_1\otimes e_1\,.$

To prove associativity we first rewrite the equation that has to be satisfied by
$F^\mathcal{A}$, namely the WDVV equation:
\[
\dfrac{\p^3F^\mathcal{A}}{\p t^{(\gamma k)}\,\p t^{(\sigma s)}\, \p t^{(\alpha i)}} {\eta}^{(\alpha i)(\beta j)}
\dfrac{\p^3F^\mathcal{A}}{\p t^{(\beta j)}\,\p t^{(\delta p)}\, \p t^{(\mu q)}}=
\dfrac{\p^3F^\mathcal{A}}{\p t^{(\mu q)}\,\p t^{(\sigma s)}\, \p t^{(\alpha i)}} {\eta}^{(\alpha i)(\beta j)}
\dfrac{\p^3F^\mathcal{A}}{\p t^{(\beta j)}\,\p t^{(\delta p)}\, \p t^{(\gamma k)}}\,.
\]
This is equivalent to
\eqa &&
\left[\widehat{c_{\gamma \sigma \alpha}}\right]^a  \omega(e_a\circ e_k\circ e_s\circ e_i)
\eta^{\alpha\beta} \eta^{ij}   \omega(e_j\circ e_p\circ e_q\circ e_b)
\left[\widehat{c_{\beta\delta\mu}}\right]^b\nn\\
&=&
\left[\widehat{c_{ \mu \sigma \alpha}}\right]^a  \omega(e_a\circ e_q\circ e_s\circ e_i) \eta^{\alpha\beta} \eta^{ij}
  \omega(e_j\circ e_p\circ e_k\circ e_b)
\left[\widehat{c_{\beta\delta\gamma}}\right]^b\,,
\nn\eeqa
which becomes, on using equation \eqref{propertyA},
\eqa &&
\left[\widehat{c_{\gamma \sigma \alpha}}\right]^a \eta^{\alpha\beta}
 \omega(e_a\circ e_k\circ e_s\circ e_p\circ e_q\circ e_b)
\left[\widehat{c_{\beta\delta\mu}}\right]^b\nn\\
&=&
\left[\widehat{c_{ \mu \sigma \alpha}}\right]^a \eta^{\alpha\beta} \omega(e_a\circ e_q\circ
e_s\circ e_p\circ e_k\circ e_b)
\left[\widehat{c_{\beta\delta\gamma}}\right]^b\,.
\label{IZ-2} \eeqa

Since the prepotential $F$ for the Frobenius manifold $\mathcal{M}$ satisfies the WDVV equation
$$ \dfrac{\p^3 F}{\p t^\gamma\,\p t^\sigma\, \p t^\alpha} {\eta}^{\alpha\beta}
\dfrac{\p^3 {F}}{\p t^\beta\,\p t^\delta\, \p t^\mu}=
\dfrac{\p^3 {F}}{\p t^\mu \,\p t^\sigma\, \p t^\alpha} {\eta}^{\alpha\beta}
\dfrac{\p^3 {F}}{\p t^\beta\,\p t^\delta\, \p t^\gamma}\,,
$$
it follows that
$$ \widehat{\dfrac{\p^3 F}{\p t^\gamma\,\p t^\sigma\, \p t^\alpha}} \circ
\widehat{ {\eta}^{\alpha\beta}} \circ \widehat{
\dfrac{\p^3 {F}}{\p t^\beta\,\p t^\delta\, \p t^\mu}}=
\widehat{\dfrac{\p^3 {F}}{\p t^\mu \,\p t^\sigma\, \p t^\alpha}} \circ
\widehat{ {\eta}^{\alpha\beta}} \circ
\widehat{\dfrac{\p^3 {F}}{\p t^\beta\,\p t^\delta\, \p t^\gamma}},$$
where $\widehat{ {\eta}^{\alpha\beta}}={\eta}^{\alpha\beta}\, e_1$.
This reduces to
\eqa
\left[\widehat{c_{\gamma \sigma \alpha}}\right]^a \eta^{\alpha\beta}
 e_a\circ e_b
\left[\widehat{c_{\beta\delta\mu}}\right]^b
=
\left[\widehat{c_{ \mu \sigma \alpha}}\right] \eta^{\alpha\beta} e_a \circ e_b
\left[\widehat{c_{\beta\delta\gamma}}\right]^b\,.\eeqa
Thus we have, by multiplying by $e_q\circ e_s\circ e_p\circ e_k$\,,
\eqa
\left[\widehat{c_{\gamma \sigma \alpha}}\right]^a \eta^{\alpha\beta}
e_a\circ e_k\circ e_s\circ e_p\circ e_q\circ e_b
\left[\widehat{c_{\beta\delta\mu}}\right]^b =
\left[\widehat{c_{ \mu \sigma \alpha}}\right]^a \eta^{\alpha\beta}
e_a\circ e_q\circ e_s\circ e_p\circ e_k\circ e_b
\left[\widehat{c_{\beta\delta\gamma}}\right]^b,\nn
\nn\eeqa
and evaluating the function with $\omega\,,$ gives the identity \eqref{IZ-2}. Hence $F^\mathcal{A}$
satisfies the WDVV equation
in the flat coordinates of the metric $\eta_{(\alpha i)(\beta j)}\,.$

\noindent{\underline{Quasi-homogeneity}}

This follows immediately from the definition of $F^\mathcal{A}$, but one can also derive the result
by direct computation. The quasi-homogeneity of $F$
is expressed by the equation
\[
\sum_\alpha \left[(1-q_\alpha) t^\alpha + r_\alpha \right] \frac{\partial F}{\partial t^\alpha} = (3-d) F
\]
where quadratic terms will be ignored. On lifting this and using the evaluation map defined by $\omega$ one obtains
\[
\sum_{(\alpha i)} (1-q_\alpha) t^{(\alpha i)} \omega
\left(\widehat{\left(\frac{\partial F}{\partial t^\alpha}\right)} \circ e_i\right) + \sum_\alpha r_\alpha
\omega\left( \widehat{\frac{\partial F}{\partial t^\alpha}} \right) = (3-d) F^\mathcal{A}\,.
\]
Using \eqref{propertyB} yields the result ${E}^\mathcal{A}\left(F^\mathcal{A}\right) = (3-d) F^\mathcal{A}$
(again, up to quadratic terms) where
\[
E^\mathcal{A} = \sum_{(\alpha i)} (1-q_\alpha) t^{(\alpha i)} \frac{\partial~}{\partial t^{(\alpha i)}}
+ \sum_\alpha r_\alpha \frac{\partial~}{\partial t^{(\alpha 1)}}\,.
\]

\noindent These show that $F^\mathcal{A}$ defines a Frobenius manifold. It remains to show that this
is the tensor product $\mathcal{M} \otimes \mathcal{A}\,.$
In fact this is straightforward. Parts (i) and (iii) of Proposition \ref{tensor} are immediate from
above (since for the trivial Frobenius manifold $\mathcal{A}$, $q_i=r_i=d=0$),
so it just remains to verify condition (ii). Since $c_{\alpha\beta}^{\phantom{\alpha\beta}\gamma}$
is independent of $t^1$ it follows that at points $t^{(\alpha i)}=0\,,\alpha>1\,,i>1$
that $\widehat{c_{\alpha\beta}^{\phantom{\alpha\beta}\gamma}}=c_{\alpha\beta}^{\phantom{\alpha\beta}\gamma}
\left(t^{(\sigma1)}\right) e_1$ and the result follows
from equation \eqref{liftedmultiplication}.

Hence the prepotential $F^\mathcal{A}=\omega(\widehat{F})$ defines the Frobenius manifold structure
on the tensor product $\mathcal{M}_\mathcal{A} = \mathcal{M}\otimes \mathcal{A}\,.$ If the multiplications on $\mathcal{M}$ and $\mathcal{A}$ are semisimple then the multiplication on $\mathcal{M}_\mathcal{A}$ is also
semisimple \cite{Kaufmann,KMK}.\end{proof}

\begin{rem} Note the existence of such a prepotential $F^\mathcal{A}$ for such a tensor product follows
 from the original work of Kaufmann, Kontsevich and Manin. However
the explicit form for such an $F^\mathcal{A}$ is not immediate from their construction. The above
result gives an explicit and easily computable prepotential in the case
when one of the manifolds is trivial.
\end{rem}

\begin{ex}Let $\mathcal{M}$ be a one-dimensional Frobenius manifold
$$F(t^1)=\dfrac{1}{6}(t^1)^3,\, e=\p_1, \, E=t^1\p_1,$$
so $\mathcal{M}_{\mathcal{A}}=\mathcal{A}$
given in Example \ref{trivial}.
\end{ex}

\begin{ex} Suppose $\mathcal{A}$ is a Frobenius algebra $\mathcal{Z}_{2,2}^{\var, 0}$
defined in Example \ref{ex2.2}.  When $\var \ne 0$, $\mathcal{A}$ is semisimple.
When $\var=0$, $\mathcal{A}$ is nonsemisimple and exactly the algebra $\mathcal{Z}_{2,2}$
given in Example \ref{ex2.3}. Let $\mathcal{M}$ be a 2-dimensional Frobenius manifold with the flat
coordinate $(t^1,t^2)$. We denote
$$\widehat{t^1}=v^1e_1+v^2 e_2,\quad \widehat{t^2}= v^3e_1+v^4 e_2.$$

{\bf Case 1}. $\mathcal{M}=\mathbb{C}^2/W(A_2)$,  i.e.,
\beq F(t)=\frac{1}{2}(t^1)^2t^2-\frac{1}{72}(t^2)^4, \quad e=\frac{\p}{\p t^1},\quad
E=t^1 \frac{\p}{\p t^1}+ \frac{2}{3}t^2 \frac{\p}{\p t^2}.\nn \eeq
The unity vector field and the Euler vector field of $\mathcal{M}_\mathcal{A}$ are given by, respectively,
$$e=\frac{\p}{\p v^1}, \quad E^\mathcal{A}=v^1\frac{\p}{\p v^1}+v^2\frac{\p}{\p v^2}
+\frac{2}{3}v^3\frac{\p}{\p v^3}+\frac{2}{3}v^4\frac{\p}{\p v^4}$$
and the potential of $\mathcal{M}_\mathcal{A}$ is given by
\beq F^\mathcal{A}(v)=\frac{1}{2}(v^1)^2v^4+v^1v^2v^3-\frac{1}{18}(v^3)^3v^4
+\var\left(\frac{1}{2}(v^2)^2v^4-\frac{1}{18}v^3(v^4)^3\right).
\nn\eeq
We remark that when $\var\ne0$, $\mathcal{M}_\mathcal{A}$ is a polynomial semisimple
Frobenius manifold. By a result of Hertling \cite{Hert}, the manifold $\mathcal{M}_\mathcal{A}$ decomposes into a product of
$A_2$-Frobenius manifolds. The algebra $\mathcal{A}$ can be seen as controlling this decomposition.

\medskip

{\bf Case 2}. $\mathcal{M}=\mathrm{QH^*(\mathrm{CP}^1)}$, i.e.,
$$F(t)=\frac{1}{2}(t^1)^2t^2+e^{t^2}, \quad e=\dfrac{\p}{\p t^1},
\quad E=t^1\dfrac{\p}{\p t^1}+2\dfrac{\p}{\p t^2}.$$
The unity vector field and the Euler vector field of $\mathcal{M}_\mathcal{A}$ are given by, respectively,
$$e=\dfrac{\p}{\p v^1},\quad E^\mathcal{A}=v^1\dfrac{\p}{\p v^1}+v^2\dfrac{\p}{\p v^2}
+2\dfrac{\p}{\p v^3}+2\dfrac{\p}{\p v^4}$$
and the potential of $\mathcal{M}_\mathcal{A}$ is given by
\beq F^\mathcal{A}(v)=\left\{\begin{array}{ll}
 \dfrac{1}{2}(v^1)^2v^4+v^1v^2v^3+\var (v^2)^2v^4+
\dfrac{\sinh(\sqrt{\var} v^4)}{\sqrt{\var}}\,e^{v^3},&\var\ne 0,\\
\dfrac{1}{2}(v^1)^2v^4+v^1v^2v^3+v^4\,e^{v^3}, &\var=0.\end{array}\right.\nn\eeq
\end{ex}


\section{$\mathcal{A}$-valued Topological Quantum Field Theories}\label{ATQFT}

The ideas developed in the last section may be applied to the construction of $\mathcal{A}$-valued
Topological Quantum Fields Theories on a suitably defined big-phase space (i.e. with gravitational descendent fields).
In fact one could have started with this larger construction
and obtained the results of the last section by restriction to the small-phase space. Conversely, the reconstruction theorems
which give big-phase space structures from Frobenius manifold structures could be used to construct these $\mathcal{A}$-valued TQFTs
from the Frobenius manifold $\mathcal{M}_\mathcal{A}\,.$

\subsection{Background}

A topological quantum field theory (or TQFT) is defined in terms of properties of certain correlators which are themselves
defined in terms of prepotential $\mathcal{F}_{g\geq 0}$\,. For example, consider a smooth projective variety $V$ with $H^{\rm
odd}(V;\mathbb{C})=0$, $\{\gamma_1\,,\ldots\,,\gamma_N\}$ a basis
for the cohomology ring $M:=H^{*}(V;\mathbb{C})$ and let
$$
\eta_{\alpha\beta}  =  \eta(\gamma_\alpha,\gamma_\beta) =  \int_V
\gamma_\alpha \cup \gamma_\beta
$$
be the Poincar\'e pairing which defines a non-degenerate metric
which may be used to raise and lower indices. Following the
conventions of Liu and Tian \cite{liu1,liu2}, a flat coordinate
system $\{t^\alpha_0\,, \alpha=1\,,\ldots\,,N\}$ may be found on
$M$ so $\gamma_\alpha=\frac{\partial~}{\partial t^\alpha_0}$, and in
which the components of $\eta$ are constant.

The big phase space consists of an infinite number of copies of
the $M\,,$ the small phase space, so
\[
M^\infty = \prod_{n\geq 0} H^{*}(V;\mathbb{C})\,.
\]
The coordinate system $\{ t_{0}^{\alpha}\}$ induces, in a canonical way,
a coordinate system $\{t^\alpha_n\,,
n \in \mathbb{Z}_{\geq 0}\,,\alpha=1\,,\ldots\,,N\}$ on $M^{\infty}.$
We denote by $\tau_n(\gamma_\alpha) = \frac{\partial~}{\partial
t^\alpha_n}$ (also abbreviated to $\tau_{n,\alpha}\,)$
the associated fundamental vector fields, which
represent various tautological line bundles over the
moduli space of curves.

The descendant Gromov-Witten invariants
\[
\langle \tau_{n_1}(\gamma_{a_1}) \ldots \tau_{n_k}(\gamma_{a_k})\rangle_g
\]
may be combined into generating functions, called prepotentials, labeled by the genus $g\,,$
\[
{\mathcal{F}}_g=\sum_{k\geq 0} \frac{1}{k!} \sum_{n_1 ,\alpha_1\ldots
n_k,\alpha_k} t^{\alpha_1}_{n_1} \ldots t^{\alpha_k}_{n_k} \langle
\tau_{n_1}(\gamma_{\alpha_1}) \ldots
\tau_{n_k}(\gamma_{\alpha_k})\rangle_g\,,
\]
and these in turn may be used to define $k$-tensor fields on the big phase space, via the formula
\begin{equation}\label{k-point}
\langle\langle {\mathcal W}_{1}\cdots {\mathcal W}_{k}\rangle\rangle_{g} = \sum_{m_{1},\alpha_{1},\cdots , m_{k},\alpha_{k}} f^{1}_{m_{1},\alpha_{1}}
\cdots f^{k}_{m_{k},\alpha_{k}} \frac{\partial^{k} {\mathcal{F}_{g}}}{\partial t^{\alpha_{1}}_{m_{1}}\cdots \partial t^{\alpha_{k}}_{m_{k}} },
\end{equation}
for any vector fields ${\mathcal W}_{i} = \sum_{m,\alpha} f^{i}_{m,\alpha} \frac{\partial}{\partial t^{\alpha}_{m}}$. The tensor field (\ref{k-point}) has a physical interpretation
as the $k$-point correlation function of the TQFT.

The basic relationships between these correlators may then be encapsulated in the following:

\begin{de}
Let $\tilde{t}^\alpha_n=t^\alpha_n - \delta_{n,1} \delta_{\alpha,1}$ and let
\begin{eqnarray*}
\mathcal{S} & = & -\sum_{n,\alpha} \tilde{t}^\alpha_n \tau_{n-1}(\gamma_{\alpha} )\,,\\
\mathcal{D} & = & -\sum_{n,\alpha} \tilde{t}^\alpha_n
\tau_{n}(\gamma_{\alpha})
\end{eqnarray*}
be the string and dilaton vector fields respectively. Then the prepotentials ${\mathcal{F}}_g$ satisfy the following relations:

\medskip
\noindent{\underline{String Equation:}}
\[
\langle\langle \mathcal{S} \rangle\rangle_g = \frac{1}{2} \delta_{g,0} \sum_{\alpha,\beta} \eta_{\alpha\beta} t^\alpha_0 t^\beta_0\,;
\]
\medskip
\noindent{\underline{Dilaton Equation:}}
\[
\langle\langle \mathcal{D} \rangle\rangle_g = (2g-2) {\mathcal{F}}_g - \frac{1}{24} \chi(V) \delta_{g,1}\,;
\]
\medskip
\noindent{\underline{Genus-zero Topological Recursion Relation:}}
\[
\langle\langle\tau_{m+1}(\gamma_\alpha) \tau_n(\gamma_\beta) \tau_k(\gamma_\sigma)  \ra\ra_{{}_0} = \la\la \tau_{m}(\gamma_\alpha)\gamma_\mu \ra\ra_{{}_0}  \la\la\gamma^\mu \tau_n(\gamma_\beta) \tau_k(\gamma_\sigma) \rangle\rangle_{{}_0}\,.
\]
\end{de}
By restricting such theories to primary vector fields with coefficients in the small phase space one recovers a Frobenius manifold structure \cite{Du,Du2} on the small phase space,
with
\[
F_0(t_0^1\,,\ldots\,,t_0^N) = \left.\mathcal{F}_0( {\bf t})
\right|_{t^\alpha_n=0\,,\,n>0}
\]
becoming the prepotential for the Frobenius manifold
and multiplication given by
$$
\tau_{0,\alpha}\circ\tau_{0,\beta} = \langle\langle\tau_{0,\alpha}\tau_{0,\beta}
\gamma^{\sigma}\rangle\rangle_{{}_0}\vert_{M} \gamma_{\sigma}.
$$

\subsection{$\mathcal{A}$-TQFT}

Given such a theory one may extend the previous construction to obtain a new TQFT. Again, the existence of such a result
follows from various reconstruction theorems, but explicit formulae may be obtained when one tensors by a constant Frobenius
algebra.

\begin{thm}\label{MthmB} Let $\mathcal{F}_{g\geq 0}$ be the prepotentials defining a TQFT, $\mathcal{S}$ and $\mathcal{D}$ the corresponding
String and Dilaton vector fields and $\mathcal{A}$ be a trivial Frobenius algebra. Let $f$ be an analytic function on $\mathcal{M}^\infty$
(that is, analytic in the flat coordinates $t_N^\alpha$ for $\mathcal{M}^\infty$) and define the $\mathcal{A}$-valued
function $\hat{f}$ to be:
\begin{equation}\label{biglift}
{\hat{f}} = \left. f\right|_{t^\alpha_N \mapsto t^{(\alpha i)}_N e_i}\,,\qquad N\in\mathbb{Z}_{\geq 0}\,,\quad\alpha=1\,,\ldots \,, m,\quad i=1\,,\ldots\,,n\,.
\end{equation}
Then the functions
\[
\mathcal{F}^{\mathcal{A}}_{g\geq 0} = \omega\left( {\widehat{ \mathcal{F}}_{g\geq 0} }\right)
\]
and vector fields
\begin{eqnarray*}
\mathcal{S}^\mathcal{A} & = & -\sum_{N,(\alpha i)} \tilde{t}^{(\alpha,i)}_N \tau_{N-1,(\alpha i)}\,,\\
\mathcal{D}^\mathcal{A} & = & -\sum_{N,(\alpha i)} \tilde{t}^{(\alpha,i)}_N \tau_{N,(\alpha i)}
\end{eqnarray*}
satisfy the axioms of a Topological Quantum Field Theory.

\end{thm}

\begin{proof}

\noindent{\underline{Genus-zero Topological Recursion Relation}}

By repeating the construction in Theorem \ref{MainA} (essentially using  \eqref{propertyB}) one easily obtains the equation
\[
\langle\langle \tau_{M+1,(\alpha i)} \tau_{N,(\beta j)} \tau_{K,(\sigma k)}\rangle\rangle_{{}_0} = \omega \left(
\langle\langle \tau_{M+1,\alpha} \tau_{N,\beta} \tau_{K,\sigma}\rangle\rangle_{{}_0}^{\hat{}}  \circ e_i \circ e_j \circ e_k \right)
\]
(where we displace the $\hat{}$ symbol for notational convenience, so $f\,{}^{\hat{}}=\hat{f}$). On using the topological recursion relation this decomposes as
\[
\langle\langle \tau_{M+1,(\alpha i)} \tau_{N,(\beta j)} \tau_{K,(\sigma k)}\rangle\rangle_{{}_0} =
\eta^{\mu\nu} \omega \left(
\langle\langle \tau_{M,\alpha} \gamma_\mu \rangle\rangle_{{}_0}^{\hat{}} \circ e_i \circ e_j \circ
\langle\langle \gamma_\mu \tau_{N,\beta} \tau_{K,\sigma} \rangle\rangle_{{}_0}^{\hat{}}\circ e_k
\right)\,
\]
\[
=
\eta^{\mu\nu} \omega \left(
\langle\langle \tau_{M,\alpha} \gamma_\mu \rangle\rangle_{{}_0}^{\hat{}} \circ e_i \circ e_r\right) \eta^{rs}
\omega \left(
e_s \circ e_j \circ \langle\langle \gamma_\mu \tau_{N,\beta} \tau_{K,\sigma} \rangle\rangle_{{}_0}^{\hat{}}\circ e_k \right)
\]
on using  \eqref{propertyA}. Since
\begin{eqnarray*}
\langle\langle \tau_{M,(\alpha i)} \gamma_{(\mu r)} \rangle\rangle_0 & = & \omega\left(\langle\langle \tau_{M,\alpha} \gamma_\mu \rangle\rangle_{{}_0}^{\hat{}} \circ e_i \circ e_r\right)\,,\\
\langle\langle \gamma_{(\mu s)} \tau_{N,(\beta j)} \tau_{K,(\sigma k)} \rangle\rangle_{{}_0} & = &
\omega \left(
e_s \circ \langle\langle \gamma_\mu \tau_{N,\beta} \tau_{K,\sigma} \rangle\rangle_{{}_0}^{\hat{}}\circ e_s \circ e_j \circ e_k \right),
\end{eqnarray*}
the result follows.

\medskip

\noindent{\underline{String Equation}}

Again, on using  \eqref{propertyB} it follows that
\begin{eqnarray*}
\langle\langle \mathcal{S}^\mathcal{A} \rangle\rangle_g & = &  - \sum_{M,(\alpha i)} {\tilde{t}}^{(\alpha,i)}_M
\omega\left[
\widehat{\frac{\partial \mathcal{F}_g}{\partial t^\alpha_{M-1}}} \circ e_i\right]\,,\\
& = & \omega\left( \langle\langle \mathcal{S} \rangle\rangle^{\hat{}}_g \right)\,.
\end{eqnarray*}
Since $\mathcal{S}$ satisfies the string equation,
\begin{eqnarray*}
\langle\langle \mathcal{S}^\mathcal{A} \rangle\rangle_g & = & \frac{1}{2} \delta_{g,0} \omega \left[ \sum_{\alpha,\beta} \hat{t^\alpha_0} \circ \hat{t^\beta_0} \right]\,,\\
&=&\frac{1}{2} \delta_{g,0} \sum_{(\alpha,i),(\beta,j)} \eta_{(\alpha i)(\beta j)} t^{(\alpha i)}_0 t^{(\beta j)}_0\,,
\end{eqnarray*}
using the definition of the lifting map and the fundamental property $\omega(e_i \circ e_j)=\eta_{ij}\,.$

\medskip

\noindent{\underline{Dilaton Equation}}

Similarly, since $\mathcal{D}$ satisfies the Dilaton equation,
\begin{eqnarray*}
\langle\langle \mathcal{D}^\mathcal{A} \rangle\rangle_g & = & \omega \left( \langle\langle \mathcal{D} \rangle\rangle_g^{\hat{}} \right)\,,\\
& = & (2g-2) \omega(\hat{\mathcal{F}}_g) - \frac{1}{24} \delta_{g,1} \chi(V) \omega(e_1)\,,\\
& = & (2g-2) \mathcal{F}^\mathcal{A}_g  - \frac{1}{24} \delta_{g,1} \chi^\mathcal{A}(V),
\end{eqnarray*}
where $\chi^\mathcal{A}(V)=\chi(V) \omega(e_1)\,.$ \end{proof}

\begin{rem} The above axioms do not include the big-phase space counterpart to the Euler vector field,
but the same ideas may be applied
if such a field exists on the original TQFT.\end{rem}

The individual prepotentials may be combined into a single $\tau$-function
\[
\tau(t^\alpha_N) = e^{\sum \hbar^{g-1} \mathcal{F}_g}.
\]
In the simplest case, when $dim \mathcal{M}=1$ this defines a specific solution of the KdV hierarchy. The full connection between such
$\tau$-functions and corresponding integrable hierarchies remains an important open problem.

Since each prepotential $\mathcal{F}_g$ lifts to prepotentials $\mathcal{F}^\mathcal{A}_g$ one may define
a corresponding $\tau$-function
\[
\tau^\mathcal{A}(t^\alpha_N) = e^{\sum \hbar^{g-1} \mathcal{F}^\mathcal{A}_g}
\]
and it is clear that $\tau^\mathcal{A} = \omega\left[\hat{\tau}\right]\,.$ It seems natural to conjecture that such
a function should define a solution to a corresponding $\mathcal{A}$-valued dispersive integrable hierarchy. However, this first
requires the development of a theory of such $\mathcal{A}$-valued hierarchies.

\subsection{The role of the Frobenius form $\omega$}

The Frobenius form $\omega$ plays a vital role in the above constructions; without it one only has $\mathcal{A}$-valued objects.
However, one can dispense with it and deal directly with such $\mathcal{A}$-valued objects and derive relations satisfied by them.
For example, using the lifting map \eqref{biglift}, one can define $\mathcal{A}$-valued \lq correlators\rq~:
\begin{eqnarray*}
\langle\langle \tau_{N,(\alpha i)} \ldots \tau_{M,(\beta j)} \rangle\rangle_g^{\mathcal A} & = & \left[
\frac{\partial~}{\partial t^\alpha_N} \ldots \frac{\partial~}{\partial t^\beta_M} \mathcal{F}_g\right]^{\hat{}} \circ e_i \circ \ldots \circ e_j\,,\\
&=& \langle\langle \tau_{N,\alpha } \ldots \tau_{M,\beta } \rangle\rangle_g^{\hat{}}\circ e_i \circ \ldots \circ e_j\,.
\end{eqnarray*}
It is straightforward to derive the following recursion relation:
\[
\Omega \circ \langle\langle \tau_{M+1, (\alpha i)} \tau_{N, (\beta j)} \tau_{K, (\sigma k)} \rangle\rangle^{\mathcal{A}}_{{}_0} \\
= \eta^{(\mu r)(\gamma s)}
\langle\langle \tau_{M,(\alpha i)} \tau_{0,(\mu r)} \rangle\rangle^\mathcal{A}_{{}_0} \circ \langle\langle \tau_{0, (\gamma s)} \tau_{N, (\beta j)} \tau_{K, (\sigma k)} \rangle\rangle^{\mathcal{A}}_{{}_0}\,,
\]
where $\Omega = \eta^{rs} e_r \circ e_s\,.$ If this element is invertible, then one can obtain a bona fide $\mathcal{A}$-valued recursion relation.
We will not further develop such a theory here.


\section{$\mathcal{A}$-valued dispersionless integrable systems}\label{dispersionlesssection}

It was shown by Dubrovin that, given a Frobenius manifold $\mathcal{M}$, one can construct
 an associated bi-Hamiltonian hierarchy of hydrodynamic type, known as the principal hierarchy, with the
geometry of the manifold encoding the various components required in its construction.
This hierarchy may be written as
\begin{equation}
\begin{array}{lcr}
\displaystyle{\frac{\partial t^\alpha}{\partial T^{(N,\sigma)}}} &
 = & \displaystyle{\mathcal{P}_1^{\alpha\beta} \, \frac{\partial h_{(N,\sigma)}}{\partial t^\beta}}\,,\\
&&\\
& = &\displaystyle{\mathcal{P}_2^{\alpha\beta} \, \frac{\partial h_{(N-1,\sigma)}}{\partial t^\beta}}
\end{array}\label{principal}
\end{equation}
with (compatible) Hamiltonian operators
\[
\mathcal{P}^{\alpha\beta}_1=\eta^{\alpha\beta} \frac{d~}{dX}\,, \qquad \mathcal{P}^{\alpha\beta}_2
= \, g^{\alpha\beta} \frac{d~}{dX} + \Gamma^{\alpha\beta}_\gamma t^\gamma_X\,,
\]
where $g^{\alpha\beta}=c^{\alpha\beta}_{\phantom{\alpha\beta}\gamma} E^\gamma$ is the intersection
 form on $\mathcal{M}$ (and $\Gamma^{\alpha\beta}_\gamma=-g^{\alpha\mu} \Gamma^\beta_{\mu\gamma}$).
The Hamiltonian densities $h_{(N,\sigma)}$ come from the coefficients in the expansion of the
deformed flat connection for the Dubrovin connection,
\[
t_\alpha(\lambda) = \sum_{N=0}^\infty h_{(N,\alpha)} \lambda^N\,,\qquad h_{(0,\alpha)}
= \eta_{\alpha\beta} t^\beta\,,
\]
and these satisfy the recursion relation
\begin{equation}\label{recursion}
\frac{\partial^2 h_{(N,\sigma)}}{\partial t^\alpha \partial t^\beta}
= c_{\alpha\beta}^{\phantom{\alpha\beta}\mu}({\bf t}) \frac{\partial h_{(N-1,\sigma)}}{\partial t^\mu}
\end{equation}
(together with certain normalization conditions).

The Frobenius manifold $\mathcal{M}_\mathcal{A}$ will automatically inherit such a
hierarchy by the very nature of it being a Frobenius manifold. However such a hierarchy is
best written as an $\mathcal{A}$-valued system, with $m$-$\mathcal{A}$-valued dependent
fields rather than $mn$-scalar-valued dependent fields.

We begin by showing how the deformed flat variables on $\mathcal{M}_\mathcal{A}$ may
be constructed from those on $\mathcal{M}\,.$ This is achieved by
lifting and evaluation the Hamiltonian densities for $\mathcal{M}\,.$

\begin{lem} Let $h_{N,\sigma}$ be the coefficients in the deformed flat connection on
$\mathcal{M}\,.$ Then the functions
\[
\mathfrak{h}_{(N,\sigma,r)} = \omega \left( \widehat{h_{(N,\sigma)}} \circ e_r \right)
\]
satisfy the recursion relation
\[
\frac{\partial^2 \mathfrak{h}_{(N,\sigma,r)}}{\partial t^{(\alpha i)} \partial t^{(\beta j)}} =
c_{(\alpha i)(\beta j)}^{\phantom{{(\alpha i)(\beta j)}}(\gamma k)}
\frac{\partial \mathfrak{h}_{(N-1,\sigma,r)}}{\partial t^{(\gamma k)}}\,.
\]
and the initial conditions $\mathfrak{h}_{(0,\sigma,r)} = \eta_{(\sigma r)(\mu s)} t^{(\mu s)}$
and hence define the deformed
flat coordinates on $\mathcal{M}_\mathcal{A}\,.$

\end{lem}

\begin{proof}

This is a straightforward calculation (we drop the $\sigma$-label on the various $h$'s for clarity):
We have
\[
\frac{\partial \widehat{h_N}}{\partial t^{(\alpha i)}}
= \widehat{\left(\frac{\partial h_N}{\partial t^\alpha}\right)} \circ e_i
\]
and hence
\begin{eqnarray*}
\frac{\partial^2 \widehat{h_N}}{\partial t^{(\alpha i)} \partial t^{(\beta j)}}
& = & \widehat{\left(\frac{\partial^2 h_N}{\partial t^\alpha \partial t^\beta}\right)} \circ e_i \circ e_j\,,\\
& = &
\widehat{\left( c_{\alpha\beta}^{~~\gamma}\right)}
\circ
\widehat{ \frac{\partial h_{N-1}}{\partial t^\gamma} } \circ e_i \circ e_j\,.
\end{eqnarray*}
Thus using $\omega$ to evaluate this $\mathcal{A}$-valued expression gives
\begin{eqnarray*}
\frac{\partial^2 \mathfrak{h}_{(N,r)}}{\partial t^{(\alpha i)} \partial t^{(\beta j)}} & = &
\omega\left(
\widehat{\left(\frac{\partial^2 h_N}{\partial t^\alpha \partial t^\beta}\right)} \circ e_i
\circ e_j \circ e_r\right)\,,\\
& = & \left[ \widehat{ c_{\alpha\beta}^{~~\gamma} }\right]^p  c_{ij}^{~~q}
\omega\left(
\widehat{ \left(\frac{\partial h_{N-1}}{\partial t^\gamma}\right) } \circ e_p \circ e_q \circ e_r
\right)\,,\\
& = & \underbrace{\left[ \widehat{ c_{\alpha\beta}^{~~\gamma} }\right]^p
 c_{ij}^{~~q} c_{pq}^{~~k} }_{c_{(\alpha i)(\beta j)}^{\phantom{(\alpha i)(\beta j)} (\gamma k)}}
\omega\left(
\frac{\partial \widehat{h_{N-1}}}{\partial t^{(\gamma k)}} \circ e_r
\right)\,,\\
& = & c_{(\alpha i)(\beta j)}^{\phantom{(\alpha i)(\beta j)} (\gamma k)}
\frac{\partial \mathfrak{h}_{(N-1,r)}}{\partial t^{(\gamma k)}}\,.
\end{eqnarray*}
If $N=0$, then, since $\widehat{t^\mu}=t^{(\mu s)} e_s\,,$
\begin{eqnarray*}
\mathfrak{h}_{(0,\sigma,r)} & = & \omega\left( \widehat{ h_{(0,\sigma)}} \circ e_r\right)\,,\\
& = & \eta_{\sigma\mu} \eta_{rs} t^{(\mu s)} \omega\left(e_s \circ e_r\right)\,,\\
& = & \eta_{(\sigma r)(\mu s)} t^{(\mu s)}\,,
\end{eqnarray*}
which is, as required, a Casimir function on $\mathcal{M}_{\mathcal{A}}\,.$
\end{proof}

In the obvious way, one can lift the operators $\mathcal{P}_1\,,\mathcal{P}_2$
to $\mathcal{A}$-valued operators and obtain the following theorem:

\begin{thm}\label{thm3.2}
The principal hierarchy on $\mathcal{M}_\mathcal{A}$ may be written in terms
of $\mathcal{A}$-valued fields, densities and operators, as
\begin{equation}
\begin{array}{lcr}
\displaystyle{
\frac{\partial {\widehat t^\alpha}}{\partial T^{(N,\sigma, r)}}} & = &
\displaystyle{\widehat{\mathcal{P}_1^{\alpha\beta}} \circ \,
\frac{\partial \widehat{h_{(N,\sigma)}}}{\partial t^{(\beta r)}}}\,,\\
&&\\
& = & \displaystyle{\widehat{\mathcal{P}_2^{\alpha\beta}} \circ \,
\frac{\partial \widehat{h_{(N-1,\sigma)}}}{\partial t^{(\beta r)}}}\,.
\end{array}
\end{equation}

\end{thm}

\begin{proof}

\noindent{\underline{First Hamiltonian Structure}}

By definition, and on using previous results,
\begin{eqnarray*}
\frac{\partial t^{(\alpha i)}}{\partial T^{(N,\sigma,r)}} & = &
\eta^{(\alpha i)(\beta j)} \frac{d~}{dX} \frac{\partial
\mathfrak{h}_{(N,\sigma,r)}}{\partial t^{(\beta j)}}\,,\\
& = & \eta^{\alpha\beta} \eta^{ij} \frac{d~}{dX}
\left[ \widehat{ \frac{\partial h_{(N,\sigma)}}{\partial t^\beta}}\right]^k \omega(e_k\circ e_j\circ e_r)\,.
\end{eqnarray*}
Since $\widehat{t^\alpha} = t^{(\alpha i)} e_i$ by definition, one obtains
\begin{eqnarray*}
\frac{\partial \widehat{t^\alpha} }{\partial T^{(N,\sigma,r)}} & = & \eta^{\alpha\beta}
\frac{d~}{dX} \left[ \widehat{ \frac{\partial h_{(N,\sigma)}}{\partial t^\beta}}\right]^k \,
\eta^{ij} \omega(e_k\circ e_j\circ e_r) e_i\,,\\
& = & \eta^{\alpha\beta} \frac{d~}{dX} \left\{ \widehat{ \frac{\partial h_{(N,\sigma)}}{\partial t^\beta}}
 \circ e_r \right\}\,,\\
& = & \widehat{\eta^{\alpha\beta}}\circ \frac{d~}{dX}
\widehat{ \frac{\partial h_{(N,\sigma)}}{\partial t^{(\beta r)}}}\,,
\end{eqnarray*}
since as the components of $\eta$ are constants, $\widehat{\eta^{\alpha\beta}} = \eta^{\alpha\beta} e_1\,.$

\noindent{\underline{Second Hamiltonian Structure}}

The second Hamiltonian operator $\mathcal{P}^{\alpha\beta}_2$ on $\mathcal{M}$ takes the form\footnote{We ignore the precise normalization of the second Hamiltonian structure. We also
assume here that the manifold $\mathcal{M}$ is non-resonant. It is easy to show that if $\mathcal{M}$ is non-resonant, then so is $\mathcal{M}_\mathcal{A}\,.$}
\[
\mathcal{P}^{\alpha\beta}_2 = g^{\alpha\beta} \frac{d~}{dX} + \left(\frac{d+1}{2} - q_\beta\right)
 c^{\alpha\beta}_{\phantom{\alpha\beta}\gamma} t^\gamma_X
\]
and hence on $\mathcal{M}_\mathcal{A}\,,$
\begin{equation}\label{second}
\frac{\partial t^{(\alpha i)}}{\partial T^{(N,\sigma,r)}} = \left[
g^{(\alpha i)(\beta j)} \frac{d~}{dX} + \left(\frac{d+1}{2} - q_\beta\right)
c^{(\alpha i)(\beta j)}_{\phantom{(\alpha i)(\beta j)}(\gamma k)} t^{(\gamma k)}_X
\right]
\frac{ \partial \mathfrak{h}_{(N-1,\sigma, r)}}{\partial t^{(\beta j)}}\,.
\end{equation}
Note, since the Euler vector field on $\mathcal{A}$ is trivial ($q_i=r_i=d_\mathcal{A}=0$)
it follows that $q_{(\beta j)} = q_\beta$ and $d$ is the same on both $\mathcal{M}$ and
$\mathcal{M}_\mathcal{A}\,.$ Also, by definition,
\begin{eqnarray*}
g^{(\alpha i)(\beta j)} & = & c^{(\alpha i)(\beta j)}_{\phantom{(\alpha i)(\beta j)}(\gamma k)} E^{(\gamma k)} \,,\\
& = & \eta^{\beta \mu} \eta^{js} \left[ \widehat{ c_{\mu\gamma}^{\phantom{\mu\gamma}\alpha} } \right]^p
c_{sk}^{\phantom{sk}q} c_{pq}^{\phantom{pq}i} (1-q_\gamma) t^{(\gamma k)}\,.
\end{eqnarray*}
For simplicity we will consider the first term in (\ref{second}) only, the corresponding proof
of the second term follows practically verbatim the proof of the first.
Thus
\begin{eqnarray*}
g^{(\alpha i)(\beta j)} \frac{d~}{dX} \frac{ \partial \mathfrak{h}_{(N-1,\sigma, r)}}{\partial t^{(\beta j)}} & = &
\left[ c^{\alpha\beta}_{\phantom{\alpha\beta}\gamma}\right]^p c_{pk}^{\phantom{pk}q} (1-q_\gamma) t^{(\gamma k)}
\frac{d~}{dX}
\left[ \widehat{ \frac{\partial h_{(N-1,\sigma)}}{\partial t^\beta}} \right]^d \omega(e_d \circ e_j \circ e_r)\,,\\
& = & \left[ \widehat{g^{\alpha\beta}} \right]^q c_q^{\phantom{q}ij} \frac{d~}{dX}
\left[ \widehat{ \frac{\partial h_{(N-1,\sigma)}}{\partial t^\beta}} \right]^d c_{dr}^{\phantom{dr}s} \eta_{sj}\,,
\end{eqnarray*}
since $\widehat{g^{\alpha\beta}} = \widehat{c^{\alpha\beta}_{\phantom{\alpha\beta}\gamma}} \circ (1-q_\gamma)
 t^{(\gamma q)} e_q\,.$ On
using the associative and commutative properties of the multiplication, and on contracting with $e_i$ one obtains
\begin{eqnarray*}
g^{(\alpha i)(\beta j)} \frac{d~}{dX} \frac{ \partial \mathfrak{h}_{(N-1,\sigma, r)}}{\partial t^{(\beta j)}} e_i
& = &
 \left[ \widehat{g^{\alpha\beta}} \right]^q
c_{qs}^{\phantom{qs}i} \frac{d~}{dX}
 \left[ \widehat{ \frac{\partial h_{(N-1,\sigma)}}{\partial t^\beta}} \circ e_r \right]^s e_i\,,\\
& = & \widehat{g^{\alpha\beta}} \circ \frac{d~}{dX}
\left[ \widehat{ \frac{ \partial h_{(N-1,\sigma)}}{\partial t^\beta}} \circ e_r \right]\,,\\
& = & \widehat{g^{\alpha\beta}} \circ \frac{d~}{dX}
\widehat{ \frac{\partial h_{(N-1,\sigma)}}{\partial t^{(\beta r)}} }\,.
\end{eqnarray*}
Note that these flows on $\mathcal{M}_\mathcal{A}$ simplify if $r=1\,.$\end{proof}

\begin{ex}
If $\dim\mathcal{M}=1$ and $r=1$ one obtains the bi-Hamiltonian structures from the $\mathcal{A}$-valued
Mong\'e equation
\[
\mathcal{U}_T= \mathcal{U}\circ \mathcal{U}_X
\]
with conserved densities
\[
\mathfrak{h}_N = \frac{1}{(N+1)!} \omega( \underbrace{\mathcal{U} \circ \cdots \circ \mathcal{U}}_{N+1\rm{~terms}} )\,.
\]

\end{ex}

The form of the flows in Theorem \ref{thm3.2} is somewhat hybrid in nature and to rewrite them as a genuine $\mathcal{A}$-valued bi-Hamiltonian
system one must introduce the variational derivative with respect to an $\mathcal{A}$-valued field. Such a derivative was introduced in \cite{OS} and is defined
by the equation
\begin{equation}\label{Avardiff}
\la \delta \mathcal{H}; v \ra = \left.\frac{d~}{d\epsilon} \mathcal{H}\left[ \widehat{u^\alpha} + \epsilon \widehat{v^\alpha} \right]\right|_{\epsilon=0},
\end{equation}
where
\[
\mathcal{H}= \int \omega(\widehat{h}) \,dX\,.
\]
With this the flows may be written as an $\mathcal{A}$-valued bi-Hamiltonian
system.

\begin{cor} \label{cor3.4} The flows given in Theorem \ref{thm3.2} may be written as
\begin{equation}
\begin{array}{lcr}
\displaystyle{
\frac{\partial {\widehat t^\alpha}}{\partial T^{(N,\sigma, r)}}} & = &
\displaystyle{\widehat{\mathcal{P}_1^{\alpha\beta}} \circ \,
\frac{\delta \mathcal{H}_{(N,\sigma,r)}}{\delta \widehat{t^\beta}}}\,,\\
&&\\
& = & \displaystyle{\widehat{\mathcal{P}_2^{\alpha\beta}} \circ \,
\frac{\delta \mathcal{H}_{(N-1,\sigma,r)}}{\delta \widehat{t^\beta}}}\,,
\end{array}
\end{equation}
where
\[
\mathcal{H}_{(N,\sigma,r)} = \int \omega\left(\widehat{h_{(N,\sigma,r)}}\right) \, dX\,.
\]
\end{cor}

\begin{proof}

From \eqref{Avardiff}\,,
\begin{eqnarray*}
\la \delta \mathcal{H}_{(N,\sigma,r)} ; \widehat{v^\beta} \ra & = & \int \omega \left( \frac{ \partial \widehat{h_{(N,\sigma,r)}}}{\partial t^{(\beta j)}} v^{(\beta j)} \circ e_r \right) \, dX\,,\\
& = &\int \omega\left(
\widehat{ \frac{\partial h_{(N,\sigma,r)}}{\partial t^\beta} }
\circ e_r \circ \underbrace{v^{(\beta j)} e_j}_{\widehat{v^\beta}} \right) \, dX\,,
\end{eqnarray*}
and hence
\[
\frac{\delta \mathcal{H}_{(N,\sigma,r)}}{\delta \widehat{t^\beta}} =
\frac{ \partial \widehat{h_{(N,\sigma)}}}{\partial t^{\beta}} \circ e_r \,.
\]
With this, the result follows immediately.
\end{proof}

\subsection{Polynomial (inverse)-metrics and bi-Hamiltonian structures}

Since all 1-dimensional metrics are flat, it follows immediately from the Dubrovin-Novikov \cite{DN} Theorem that the operator
\[
{\mathcal{P}} = f(u) \frac{d~}{dX} + \frac{1}{2} f^\prime(u)
\]
is Hamiltonian. In this section we study the case where $f$ is a polynomial.

\begin{ex}
Applying the lifting procedures to the operator $\mathcal{P}$ defined by the linear function $f(u)=u+\lambda$ results in the linear
operator
\begin{equation}
{\mathcal{P}^{ij}} = \left\{  c^{ij}_k u^k_X  \frac{d~}{dX} + \frac{1}{2} c^{ij}_k u^k_X \right\} + \lambda  \frac{d~}{dX}
\label{basicBNresult}
\end{equation}
defined on the Frobenius algebra $\mathcal{A}\,.$
\end{ex}

\noindent This is the Hamiltonian operator first constructed by Balinski and Novikov \cite{BN}. Similarly, more complicated examples by be obtains
by starting with more general polynomials and applying the same procedure.

These more general examples appear to be in contradiction to an alternative method of constructing Hamiltonian operators via
bi-Hamiltonian recursion. The recursion operator constructed from the bi-Hamiltonian pencil (\ref{basicBNresult}) takes the form
\[
{\mathcal{R}}^i_j = c^i_{jk} u^k + \frac{1}{2} c^i_{jk} u^k_X \left( \frac{d~}{dX} \right)^{-1}\,.
\]
Suppose one has a (local) Hamiltonian operator
\[
\mathcal{P}_n = g^{ij}_{(n)}(u) \frac{d~}{dX} + \Gamma^{ij}_{(n)k}(u) u^k_X
\]
with $g^{ij}_{(0)} = \eta^{ij}\,, \Gamma^{ij}_{(0)k}=0\,.$ Applying the operator $\mathcal{R}$ gives
\[
\left(\mathcal{R}\mathcal{P}_{(n)}\right)^{ij} = \left\{ g^{ij}_{(n+1)}(u) \frac{d~}{dX} + \Gamma^{ij}_{(n+1)k}(u) u^k_X\right\} + {\rm{non-local~terms}}
\]
and we now {\sl define} $\mathcal{P}_{(n+1)}$ to be the local-term in the above expression. This gives the recursion scheme:
\begin{eqnarray*}
g^{ij}_{(n+1)} & = & 2 c^{ip}_r u^r \eta_{pq} g^{qj}_{(n)}\,,\\
\Gamma^{ij}_{(n+1)k} & = & 2 c^{ip}_r u^r \eta_{pq} \Gamma^{qj}_{(n)k} + c^{ip}_k \eta_{pq} g^{qj}_{(n)}\,.
\end{eqnarray*}
It is a tedious, through straightforward exercise to show that, if the pair $\{g_{(n)},\Gamma_{(n)}\}$ defines a flat metric, then so does $\{g_{(n+1)},\Gamma_{(n+1)}\}$,
and hence $\mathcal{P}_{(n)}$ is a local Hamiltonian operator for all $n\,.$ The above lifting procedure circumvents such a direct computational approach. The fact that
the local and non-local parts of the Hamiltonian operator define separate, compatible, Hamiltonian operator is of course, well known (see, for example, \cite{GLR}).

\section{$\mathcal{A}$-valued dispersive integrable systems}\label{dispersivesection}

In this section the above ideas are extended to include dispersive, higher-order, dispersive systems.

\subsection{$\mathcal{A}$-valued dispersive integrable systems}\label{dispersivesectionA}

The main result of this section is the following theorem:

\begin{thm}
 Let $u=\{u^\alpha(x,t)|\alpha=1,\cdots,n \}$. Let
\beq u^\alpha_t=K^\alpha(u,u_x,\cdots) \label{eq2.4} \eeq
be a Hamiltonian system with the Hamiltonian $H[u]$, then the corresponding $\mathcal{A}$-valued system
 \beq \widehat{u^\alpha_t}=\widehat{K^\alpha(u,u_x,\cdots)} \label{eq2.5}\eeq
is also Hamiltonian with the Hamiltonian $\mathcal{H}[\widehat{u}]=\trf \left(\widehat{H[u]}\right)$.\end{thm}

\begin{proof}The proof is very similar to those done in section \ref{dispersionlesssection}.
Without loss of generality, we assume that the system \eqref{eq2.4} can be written as
\beq u^\alpha_t=\{u^\alpha, H[u]\}=\mathcal{P}^{\alpha\beta}\dfrac{\delta h}{\delta u^\beta}, \quad H[u]=\int h(u)dx,\eeq
where $\mathcal{P}^{\alpha\beta}$ is a Hamiltonian operator.  So the system \eqref{eq2.5} reads
\beq \widehat{u^\alpha_t}=\widehat{\mathcal{P}^{\alpha\beta}}\circ\widehat{\dfrac{\delta h}{\delta u^\beta}}.
\label{eq2.7}\eeq

Let
\beq \mathcal{H}[\widehat{u}]=\int \mathfrak{h}(\widehat{u})dx,\quad
\mathfrak{h}(\widehat{u})=\trf \left(\widehat{h(u)}\right).\eeq
With respect to an $\mathcal{A}$-valued field, the variational derivative
$\dfrac{\delta \mathfrak{h}}{\delta \widehat{u^\beta}}$ is defined by the formula, essentially due to \cite{OS},
\beq
\trf \int \left(\dfrac{\delta \mathfrak{h}}{\delta \widehat{u^\beta}}\circ \widehat{\delta u^\beta} \right)dx
 = \left.\frac{d}{d\epsilon}\right|_{\epsilon=0}
\mathcal{H}\left[ \widehat{u^\beta} + \epsilon \widehat{\delta u^\beta} \right].
\eeq
Observe that
\eqa && \left.\frac{d}{d\epsilon}\mathcal{H}\left[ \widehat{u^\beta}
+ \epsilon \widehat{\delta u^\beta} \right]\right|_{\epsilon=0}=
\left.\frac{d}{d\epsilon}\right|_{\epsilon=0}\trf \left(\int h(\widehat{u^\beta}+ \epsilon\widehat{\delta u^\beta}) dx\right)\,,\\
&=&\trf \left(\widehat{\left.\frac{d}{d\epsilon}\right|_{\epsilon=0} H[{u^\beta}+ \epsilon{\delta
u^\beta}]}\right)=\trf\left(\int\left(\widehat{\dfrac{\delta h}{\delta u^\beta}}\circ \widehat{\delta u^\beta}\right)dx \right) \nn \eeqa
from which follows
\beq \dfrac{\delta \mathfrak{h}}{\delta \widehat{u^\beta}}=\widehat{\dfrac{\delta h}{\delta u^\beta}}.\label{eq2.11}\eeq

For two functionals
\beq \mathcal{F}[\widehat{u}]=\int \mathfrak{f}(\widehat{u})dx,\quad
\mathcal{G}[\widehat{u}]=\int \mathfrak{g}(\widehat{u})dx,\eeq
with $\mathfrak{f}(\widehat{u})=\trf \left(\widehat{f(u)}\right)$ and $\mathfrak{g}(\widehat{u})=\trf \left(\widehat{g(u)}\right)$,
we define a bilinear bracket as
\beq \left\{\mathcal{F}[\widehat{u}], \mathcal{G}[\widehat{u}]\right\}_{\mathcal{A}}=\trf \left(\int
\dfrac{\delta \mathfrak{g}}{\delta \widehat{u^\alpha}}\circ  \widehat{\mathcal{P}^{\alpha\beta}} \circ
\dfrac{\delta \mathfrak{g}}{\delta \widehat{u^\beta}} dx\right).\label{eq2.13}\eeq
By using the definition of the hat map and \eqref{eq2.11}, we rewrite the bracket \eqref{eq2.13} as
\beq\left\{\mathcal{F}[\widehat{u}], \mathcal{G}[\widehat{u}]\right\}_{\mathcal{A}}=\trf
\widehat{\{F[u], G[u]\}},\eeq
where $F[u]=\int f(u)dx$ and $G[u]=\int g(u)dx$. Consequently, we conclude that
the bracket $\{~,~\}_{\mathcal{A}}$ is also a Poisson bracket. Furthermore using
\eqref{eq2.11}, the system \eqref{eq2.7} could be written as
\beq u^{(\alpha,i)}_t=\{u^{(\alpha,i)}, \mathcal{H}[\widehat{u}]\}_{\mathcal{A}},\quad
\mathcal{H}[\widehat{u}]=\int \trf \left(\widehat{h(u)}\right) dx.\nn\eeq
We thus complete the proof of the theorem.\end{proof}

\begin{cor}The $\mathcal{A}$-valued version of the Hamiltonian system
$u^\alpha_t=\{u^\alpha, H[u]\}$ is also Hamiltonian and given by
 \beq u^{(\alpha,i)}_t=\{u^{(\alpha,i)}, \mathcal{H}[\widehat{u}]\}_{\mathcal{A}},\quad
\mathcal{H}[\widehat{u}]=\trf\left(\widehat{H[u]}\right).\nn\eeq
\end{cor}
These results extend naturally to the lifts of bi-Hamiltonian structures, yielding $\mathcal{A}$-valued bi-Hamiltonian operators.

\subsection{mKdV and (modified)-Camassa-Holm bi-Hamiltonian structures}

The celebrated Muira transformation maps the second Hamiltonian operator of the KdV hierarchy to constant form. Explicitly, if
\[
\mathcal{H}^{KdV}_1 = -D\,, \qquad \mathcal{H}^{KdV}_2 = -D^3 + 2 u D + u_X
\]
(in this section we write $D$ in place of $\frac{d~}{dX}$). Then applying the Miura map $u=-v_X+\frac{1}{2} v^2$ gives
\[
\mathcal{H}^{KdV}_2 = \mathcal{H}^{mKdV}_1=D\,,
\]
and the second mKdV structure is then obtained by applying the same map to the third KdV Hamiltonian structure defined by bi-Hamiltonian
recursion ($\mathcal{H}_3=\mathcal{H}_2 \mathcal{H}^{-1} \mathcal{H}_2$), yielding the non-local operator
\[
\mathcal{H}_2^{mKdV} = D^3 -D v D^{-1} v D\,.
\]
Just as the Balinski-Novikov structures on the Frobenius algebra $\mathcal{A}$ may be obtained by lifting, so $\mathcal{A}$-valued non-local operators
may be found by using the above results.

\begin{prop}
The $\mathcal{A}$-valued operators, defined by lifting $\mathcal{H}^{mKdV}_1$ and $\mathcal{H}^{mKdV}_2$ to the Frobenius algebra $\mathcal{A}$ are:
\begin{eqnarray*}
\left(\mathcal{H}^{mKdV}_1\right)^{ij} & = & \eta^{ij} D\,,\\
\left(\mathcal{H}^{mKdV}_2\right)^{ij} & = & \eta^{ij} D^3 - c^{ij}_p c^p_{mn} D v^m D^{-1} v^n D\,.
\end{eqnarray*}
These may also be obtained using the $\mathcal{A}$-valued Miura map
\[
u = -v_x + \frac{1}{2} v \circ v\,.
\]
\end{prop}

\begin{proof} These results follow directly by applying the results in section \ref{dispersivesectionA}. They may also be obtained by direct (but tedious) calculation. The form of the $\mathcal{A}$-valued
Miura map is obvious, and again can be verified by direct calculations. While not developed here, one should be able to applying lifting results directly to scalar-Miura maps, with all the
actions commuting.
\end{proof}

\noindent $\mathcal{A}$-valued KdV and mKdV equations can now easily be constructed, the KdV examples coinciding with the examples constructed in \cite{SS}. Here we construct $\mathcal{A}$-valued modified Camassa-Holm equations.

\begin{ex} One may apply the standard tri-Hamiltonian \lq tricks\rq~\cite{F} to obtain the $\mathcal{A}$-valued bi-Hamiltonian pair:
\begin{eqnarray*}
\mathcal{C}^{ij}_1 & = & \eta^{ij} (D^3+D)\,,\\
\mathcal{C}^{ij}_2 & = & c^{ij}_p c^p_{mn} D v^m D^{-1} v^n D\,.
\end{eqnarray*}
Starting with the lifted Casimir of the scalar operator $\mathcal{C}_1$ one obtains the multi-component modified Camassa-Holm equation
\begin{eqnarray*}
v_T + v_{XXT} & = &\phantom{+} \frac{1}{2} v_{XXX} \circ v_X \circ v_X + v_{XX} \circ v_{XX} \circ v_x\\
&&  + \frac{1}{2} v_{XXX} \circ v \circ v + 2 v_{XX} \circ v_X \circ v + \frac{1}{2} v_X \circ v_X \circ v_X\\
&& + \frac{3}{2} v_X \circ v \circ v\,.
\end{eqnarray*}

\end{ex}

\noindent Note we use the adjective \lq modified\rq~in the original, strict, sense of equations obtained from an original, unmodified, equation via the
action of a Miura map, rather than in the looser sense of just modifying \lq by-hand\rq~the terms that appear in the equation. Two-component
examples may easily be found using one of the algebras constructed in example \ref{ex2.2}\,.

\section{Conclusions}\label{section5}

Central to the results of this paper is the use of a distinguished coordinate system, namely the flat coordinates of the Frobenius
manifold $\mathcal{M}\,.$ But the lifting procedure may be applied to any geometric structure which is analytic in some fixed
coordinate system. However, such results loose some of their coordinate free character: one is using a specific coordinate system
to define new objects then relying in their tensorial properties to define then properly in an arbitrary system of coordinates.
As an example of this, one can apply the idea to $F$-manifolds defined by Hertling and Manin \cite{HM}.

\begin{prop}
Consider an $F$-manifold with structure functions $c_{\alpha\beta}^{\phantom{\alpha\beta}\gamma}(t)$ analytic in the coordinates $\{ t^\alpha \}\,.$
Let $\mathcal{A}$ be an arbitrary Frobenius algebra. Then the structure functions defined by the lifted multiplication \eqref{liftedmultiplication}
\[
c_{(\alpha i)(\beta j)}^{\phantom{{(\alpha i)(\beta j)}}(\gamma k)}
= \left[ \widehat{c_{\alpha\beta}^{~~\gamma}} \right]^p \, c_{ij}^{~~q} c_{pq}^{~~k}
\]
define an $F$-manifold.
\end{prop}

\noindent The proof is straightforward and will be omitted. The link between $F$-manifolds and equations of hydrodynamic type has been explored by a number of
authors \cite{St2,Lorenzoni} so one should be able to apply the idea of this paper to construct their $\mathcal{A}$-valued counterparts.

\medskip

In quantum cohomology, the tensor product of Frobenius manifolds generalizes the classical K\"unneth product formula. In singularity theory it corresponds
to the direct sum of singularities. If one of the manifolds is trivial then this descriptions degenerates - there is no parameter space of versal deformations.
However, one could try to construct an $\mathcal{A}$-valued singularity theory. This is purely speculative, but Arnold has constructed a theory of versal
deformations of matrices \cite{Arnold} but it remains to see if this is what would be required.

\medskip

As remarked earlier, since $\mathcal{M}_\mathcal{A}$ is a Frobenius manifold in its own right, one can apply the deformation theory developed
by Dubrovin and Zhang \cite{DuZh} directly to the hydrodynamic flows given in Theorem \ref{thm3.2}. But central to this approach is the existence of
a single $\tau$-function. However the deformations/dispersive systems constructed in sections \ref{dispersionlesssection} and \ref{dispersivesection} have $\mathcal{A}$-valued $\tau$-functions.
Thus we have two distinct deformation procedures, unless they are connected by some set of transformations. It may be possible to construct a deformation
theory along the lines of \cite{DuZh} but with an $\mathcal{A}$-valued $\tau$-function.

\medskip

This paper has concentrated on Frobenius algebra-valued integrable systems, via their Hamiltonian structure. Other approaches to integrability - the structure of $\mathcal{A}$-valued Lax
equations, for example, have not been considered here. Part of such a theory have been constructed by the authors in \cite{SZ} where an $\mathcal{A}$-valued KP hierarchy is constructed
via such $\mathcal{A}$-valued Lax equations and operators. In a different direction there are many other algebra valued generalizations of KdV equation,
from Jordan algebra to Novikov algebra-valued systems \cite{St3,SS,Sv,Sv2}. Whether such algebra-valued systems can be combined with the theory of Frobenius manifolds remains an open question.
Developing a theory
which encompasses the non-commutative/non-local hierarchies, such as the original matrix KdV equation (\ref{matrixKdV}) would be of considerable interest and would encompass the theory developed
in this paper.


\medskip

\noindent{\bf Acknowledgments.} D.\,Zuo is grateful to Professors Qing Chen,
Yi Cheng and Youjin Zhang for constant supports and also thanks the University of
Glasgow for the hospitality.
The research of D.\,Zuo is partially supported by
NSFC(11271345, 11371338), NCET-13-0550 and the Fundamental Research
Funds for the Central Universities.


\section*{Appendix}

The lifting operation (Definition \ref{LiftingDef}) was defined only for analytic functions. However this may be extended to a wider class of functions, in particular rational functions.
This observation is based on the following:

\begin{lem}\label{lem2.4} A generic element $\kappa \in \mathcal{A}$ is invertible. \end{lem}

\begin{proof}By a similar argument laid out in \cite{Du2}, the Frobenius algebra $\mathcal{A}$
is isomorphic to orthogonal direct sum of a semi-simple and a nilpotent algebra,
\[
\mathcal{A}=\mathcal{A}_s \oplus \mathcal{A}_n
\]
with $\mathcal{A}_s$ having a basis $\pi_1,\cdots,\pi_s\,,$ with $\pi_j\circ\pi_l=\delta_{jl}\pi_j\,.$
Suppose the unity element of the algebra takes the form
\[
e=\sum_{i=1}^s a_i \pi_i + n
\]
where $n\in\mathcal{A}_n$ and so $n^N=0\,.$ Then $n \circ \pi_i = (1-a_i) \pi_i$ and hence $(1-a_i)^N \pi_i = 0\,.$ Thus $a_i=1$ and $n \circ \pi_i=0$. Since $e=e^N$ it follows that the unity element takes the form
\[
e=\sum_{i=1}^s \pi_i\,.
\]
Writing a generic element $\kappa\in\mathcal{A}$ as $\kappa = \pi + \mu$ (with $\pi\in\mathcal{A}_s\,,\mu\in\mathcal{A}_n$) then $(\kappa-\pi)^N=0$ for some $N\,.$ Expanding
this yields
\[
\pi^N = \kappa \circ \Xi(\kappa,\pi)
\]
for some funtion $\Xi\in\mathcal{A}\,.$ Since $\pi^N\in\mathcal{A}_s$ is invertible (generically) it follow that
\[
\kappa \circ \left\{ \Xi(\kappa,\pi) \circ \pi^{-N} \right\} = e\,.
\]
Hence the result.
\end{proof}



\end{document}